\documentclass[twocolumn]{autart}    

\usepackage{graphicx}          

\let\theoremstyle\undefined
\usepackage{amsthm}
\usepackage{amsmath,amssymb}
\usepackage{mathtools}
\usepackage{graphicx}
\usepackage{paralist}

\usepackage{bm}

\usepackage[colorlinks=true, allcolors=blue]{hyperref}
\usepackage{subcaption}

\usepackage{tikz}
\usetikzlibrary{automata,positioning}
\usepackage{acronym}

\usepackage{tabularx}

\usepackage{xcolor}

\usepackage{changes}
\usepackage{textcomp}
\usepackage{algorithm}
\usepackage[noend]{algpseudocode}
\usepackage[english]{babel}
\usepackage{enumerate}
\usepackage{mathrsfs}
\usepackage{mathtools}
\usepackage{multicol}

\newcommand*{\gobble}[1]{}

\usepackage{cite}
\usepackage{bbm}

\usepackage[framemethod=TikZ]{mdframed}

\acrodef{mdp}[MDP]{Markov decision process}
\acrodef{pomdp}[POMDP]{Partially Observable Markov Decision Process}
\acrodef{momdp}[MOMDP]{Multi-objective MDP}
\acrodef{dfa}[DFA]{deterministic finite automaton}
\acrodef{tlmdp}[TLMDP]{terminating labeled Markov decision process}
\acrodef{lmdp}[LMDP]{labeled Markov decision process}
\acrodef{pdfa}[PDFA]{preference deterministic finite automaton}
\acrodef{pdra}[PDRA]{preference deterministic Rabin automaton}
\acrodef{cpltlf}[CPLTL$_f$]{Conditional Preference over LTL$_f$}
\acrodef{cpa}[CPA]{Conditional Preference Automaton}
\acrodef{ltl}[LTL]{linear temporal logic}
\acrodef{ltlf}[LTL$_f$]{linear temporal logic over finite traces}
 \newtheorem{corollary}{Corollary}
\newtheorem{lemma}{Lemma}
\newtheorem{proposition}{Proposition}
\newtheorem{theorem}{Theorem}

\newtheorem{remark}{Remark}

\theoremstyle{definition}

\newtheorem{definition}{Definition}
\newtheorem{problem}{Problem}

\newcommand{\refDef}[1]{Def.~\ref{#1}}

\newcommand{\refFig}[1]{Fig.~(\ref{#1})}
\newcommand{\refLma}[1]{Lma.~\ref{#1}}

\newcommand{\refProp}[1]{Proposition~\ref{#1}}

\newcommand{\refThm}[1]{Thm.~\ref{#1}}

\newcommand{\ak}[1]{\textcolor{magenta}{[AK: #1]}}

\newcommand{\ie}{i.e.}

\newcommand{\prefltlf}{PrefLTL$_f$}

\newcommand{\lang}{\mathcal{L}}

\newcommand{\rank}{\mathsf{rank}}
\newcommand{\last}{\mathsf{Last}}

\newcommand{\calA}{\mathcal{A}}
\newcommand{\calE}{\mathcal{E}}

\newcommand{\calG}{\mathcal{G}}

\newcommand{\truev}{\mathsf{true}}

\newcommand{\Always}{\Box \, }
\newcommand{\Eventually}{\Diamond \, }
\newcommand{\Next}{\bigcirc \, }
\newcommand{\until}{\mbox{$\, {\sf U}\,$}}

\newcommand{\weakpref}{\trianglerighteq}
\newcommand{\strictpref}{\triangleright}
\newcommand{\indifferent}{\sim}

\DeclareMathOperator*{\prefAnd}{\&}

\newcommand{\Paths}{\mathsf{Path}}

\newcommand{\trace}{L}

\newcommand{\reach}{\mathsf{Reach}}
\newcommand{\swin}{\mathsf{SWin}}

\newcommand{\pa}{\mathcal{P}}

\newcommand{\maxRank}{\mathsf{MaxRank}}
\newcommand{\maximal}{\mathsf{Max}}

\begin{document}

\begin{frontmatter}

\title{Nash Equilibrium in Games on Graphs with Incomplete Preferences} 


\author[UT]{Abhishek N. Kulkarni}\ead{abhishek.kulkarni@austin.utexas.edu},    
\author[UF]{Jie Fu}\ead{fujie@ufl.edu},               
\author[UT]{ Ufuk Topcu}\ead{utopcu@utexas.edu}  

\address[UT]{University of Texas at Austin, Austin, TX, USA}  
\address[UF]{University of Florida, Gainesville, FL, USA}             


\begin{abstract}                          
Games with incomplete preferences are an important model for studying rational decision-making in scenarios where players face incomplete information about their preferences and must contend with incomparable outcomes.
We study the problem of computing Nash equilibrium in a subclass of two-player games played on graphs where each player seeks to maximally satisfy their (possibly incomplete) preferences over a set of temporal goals. 
We characterize the Nash equilibrium and prove its existence in scenarios where player preferences are fully aligned, partially aligned, and completely opposite, in terms of the well-known solution concepts of sure winning and Pareto efficiency. 
When preferences are partially aligned, we derive conditions under which a player needs cooperation and demonstrate that the Nash equilibria depend not only on the preference alignment but also on whether the players need cooperation to achieve a better outcome and whether they are willing to cooperate.
We illustrate the theoretical results by solving a mechanism design problem for a drone delivery scenario.
\end{abstract}

\end{frontmatter}

\section{Introduction}
\label{sec:introduction}

Games with incomplete preferences model strategic interactions where players aim to maximally satisfy their preferences, \ie, achieve the best possible outcome for themselves, by strategically responding to the other player's actions. 
In contrast to the widely studied games that assume complete preferences \cite{baier2008planning,tumova2013least, wongpiromsarn2021, rahmani2020what}, games with incomplete preferences enable the players to make rational decisions in situations involving incomparability, where decisions must be made even when some outcomes cannot be ranked \cite{bade2005nash}. 

An important problem in the study of games with incomplete preferences is the characterization of Nash equilibrium \cite{bade2005nash}, which is a set of strategies where no player can benefit by unilaterally changing their strategy \cite{shapley1953stochastic}. 


We study the problem of characterizing Nash equilibria in a subclass of games called deterministic two-player turn-based games on graphs \cite{gradel2003automata} where the player preferences are defined over a set of temporal goals expressed as \ac{ltlf} formulas \cite{de2013linear}.
A game on graph is a widely studied model for sequential decision making, particularly useful for specifying, verifying, as well as synthesizing the behavior of reactive systems \cite{chatterjee2007algorithms}. 
While the characterization of Nash equilibrium in games on graphs has been studied for complete and lexicographic preferences \cite{bade2005nash,chatterjee2023stochastic}, this problem remains underexplored for games on graphs with incomplete preferences, which is a broader class of preferences that includes both complete and lexicographic preferences.

Computing Nash equilibria in a game on graph with incomplete preferences poses two key challenges.
First, planning with preferences over temporal goals requires the agents to simultaneously evaluate their ability to satisfy various subsets of temporal goals \cite{fu2021probabilistic,kulkarni2022opportunistic}. 
This requirement arises because, in a game on graph, an outcome (\ie, a play in the game) may satisfy multiple temporal goals. 
For example, given two temporal goals, ``go to the kitchen'' and ``go to the living room,'' a game play that first visits the kitchen and then visits living room satisfies both goals.
Consequently, it might be possible for a player to assist the other player in achieving a better outcome while still satisfying the best possible outcome for themselves. 
In other words, each player needs a way to determine \emph{when should they cooperate with the other player?}

Second, the incomparability between outcomes makes the existing approaches \cite{chatterjee2023stochastic} unsuitable for rational decision making in games on graphs with incomplete preferences.
Recent works have studied the problem of rational decision making in games on graph in presence of incomparability for single-agent planning \cite{kulkarni2022opportunistic,rahmani2023probabilistic}.
However, this problem has not been studied for games on graphs containing two or more players.



%

\textbf{Contributions.} 
We characterize the set of Nash equilibria based on whether the player preferences are fully aligned, entirely opposite, or partially aligned.
This classification is motivated by the study of multi-agent systems, where the player behaviors are categorized as cooperative, semi-cooperative, or competitive \cite{shoham2008multiagent}.
In case of preferences, fully aligned preferences results in the players being fully cooperative, completely opposite preference results in them to be fully competitive, whereas the partially aligned case motivates semi-cooperation.
We describe our contributions for each of these cases separately.



In the fully aligned case, 
we follow an automata-theoretic approach to define a product of a game on graph with preference automata \cite{rahmani2023probabilistic}.
The product transforms the player preferences over temporal goals to preferences over the states of the product game.
In this product game, we show that a pair of P1 and P2 strategies is a Nash equilibrium if and only if it induces a visit to a maximal reachable state, \ie, a most-preferred state that is reachable from the initial state.

In the completely opposite case, 
incomparability poses a significant challenge in determining the best possible outcome a player can achieve against any counter-strategy of the opponent. 
The standard approach to rational decision-making under incomparability relies on the undominance principle \cite{sen1997maximization}. However, there is currently no method available to compute an undominated strategy in a two-player game on a graph.
To this end, we introduce the concept of maximal sure winning for a player based on a weaker form of undominance.  
Maximal sure winning extends the solution concept of sure winning \cite{de2007concurrent}, which is defined for games on graphs, to satisfy the best possible temporal goal against a competitive opponent.  
We show that a pair of P1 and P2 strategies is a Nash equilibrium if and only if both strategies are maximal sure winning for the respective players.

The analysis of the completely opposite case provides an important insight.
Even though the automata-theoretic approach enables reasoning about outcomes by examining the states of the product game, computing Nash equilibria requires an in-depth analysis of paths in the product game.
Specifically, it involves determining if a player has a strategy that guarantees a path from a certain subset of paths against all potential strategies of the opponent.
This approach contrasts with common methods in the literature, such as \cite{zlotkin1990negotiation} that rely on operations purely on the set of states to determine Nash equilibria.


In the partially aligned case, players may be have an incentive to cooperative depending on the degree of alignment of their preferences.
We categorize player attitudes as either agnostic or cooperative.
An agnostic player disregards the preferences of the other player in its decision-making and tries to achieve the best possible outcome for itself. 
In this case, we propose a fixed-point polynomial-time algorithm to compute the set of Nash equilibria. 
On the other hand, a cooperative player adjusts its strategy to help the other player achieve a better outcome while also achieving the best possible outcome for itself.
In a game containing cooperative players, we derive conditions under which a player needs cooperation from the other to achieve a strictly preferred outcome.
When at most one player requires cooperation, we demonstrate that the Nash strategy for the player who does not require cooperation is a maximal sure winning strategy. 
However, this strategy limits the possible outcomes to a subset of all outcomes under that strategy, maximizing the satisfaction of the other player's preferences.
Whereas, when both players need cooperation, we show that the set of Nash equilibrium is equal to the set of Pareto equilibrium \cite{wang1993existence} unless one player has a maximal sure winning strategy that achieves an outcome strictly preferred outcome to a Pareto one.
In such a case, the set of Nash equilibria can be determined using the approach described for the case when at most one player requires cooperation.

Additionally, our characterization of Nash equilibria in various cases establishes that the set of Nash equilibria is non-empty in every deterministic two-player turn-based game on graph with incomplete preferences.

Our results are particularly useful for designing games that motivate desired behavior from players given incomplete preferences, as we demonstrate using a drone delivery scenario.
They also provide the basis for studying stochastic games with incomplete preferences, which provide key insights into rational behavior of players in various applications such as robotics \cite{mitsunaga2008adapting,adinolf2020my,correia2019choose}, economics \cite{duffie1986equilibrium,duffie1987stochastic}, and social networks \cite{salehi2015preference,rezaei2024social}. 


\textbf{Related Work. }
The study of games with preferences presents unique challenges that are not adequately addressed by traditional game theory models.
The seminal work \cite{bade2005nash} on normal-form games with preferences characterized the Nash equilibria in these games as the union of Nash equilibrium in all games where player objectives are a completion of their incomplete preferences.
However, the results in \cite{bade2005nash} require strong separability conditions \cite{bosi2012continuous,evren2011incomplete} for the completions to exist, thus limiting the applicability of this approach.


Recently, posetal games \cite{zanardi2022posetal}\footnote{Posetal games should not be confused with either poset games \cite{soltys2011complexity,byrnes2002poset} or partial order games \cite{zahoransky2021partial}. Poset games are two-player impartial combinatorial games where, on each move, a player picks an element $x$ from a set $A$ and removes all elements in $A$ that are greater than or equal to $x$, forming a smaller poset $A'$. The game continues with the other player making the next move, and the player who cannot make a move (when $A = \emptyset$) loses. On the other hand, partial order games are non-cooperative game models where players' decision nodes are partially ordered by a dependence relation, directly capturing informational dependencies in the game. Both of these formulations are distinct from the approach considered in this paper.} have demonstrated the existence of Nash equilibrium in scenarios where players express preferences over outcomes using a partially ordered set of metrics. However, posetal games are studied only for normal-form games. In contrast, we investigate a reactive game between two players, where players choose their actions based on the entire history of their interactions.
Within the class of reactive games with preferences, \cite{chatterjee2023stochastic} presents a quantitative analysis of a stochastic game with lexicographic objectives by computing optimal strategies through a sequence of single-objective games. Unlike \cite{chatterjee2023stochastic}, we propose a qualitative solution applicable to arbitrary preferences over temporal goals, that subsume lexicographic preferences.

\section{Preliminaries}
\label{sec:preliminaries}
\subsection{Interaction Model}

\begin{definition}
	\label{def:dtptb-game}
    A deterministic two-player turn-based game on graph is a tuple,
    \[
        G = (S, A, T, AP, L)
    \]
    where
    \begin{inparaenum}[]
        \item $S = S_1 \cup S_2$ is a set of states. $S_1$ is the set of P1 states and $S_2$ is the set of P2 states. $S_1$ and $S_2$ are disjoint sets.
        \item $A = A_1 \cup A_2$ is a set of actions, where $A_1, A_2$ represent the action sets P1 and P2, respectively.
        \item $T: S \times A \rightarrow S$ is a deterministic transition function.
        \item $AP$ is a set of atomic propositions.
        \item $L: S \rightarrow 2^{AP}$ is a labeling function.
    \end{inparaenum}
\end{definition}

A \emph{path} in $G$ is a (finite/infinite) sequence of states $\rho = s_0 s_1 s_2 \cdots$ such that, for every $i \geq 0$, there exists an action $a_i \in A$ such that  $s_{i+1} = T(s_i, a_i)$.
The path is said to be \emph{finite} if it terminates after a finite number of steps, otherwise it is \emph{infinite}.
We denote the set of all finite paths in $G$ by $\Paths(G)$ and that of infinite paths is denoted by $\Paths^\infty(G)$.
A finite path $\rho$ induces a finite word $\trace(\rho) = L(s_0) L(s_1) L(s_2) \cdots L(s_n) \in (2^{AP})^*$ called the \emph{trace} of $\rho$.
The last state of a finite path $\rho$ is denoted by $\last(\rho)$.
The trace of an infinite word is defined analogously.

A \emph{finite-memory}, \emph{set-based} strategy for player-$i$, $i = 1, 2$, in $G$ is a function $\pi_i: S^+ \rightarrow 2^{A}$, which maps every finite path in $\Paths(G)$ to a subset of actions.
The set of all finite-memory, set-based strategies of player-$i$ in $G$ is denoted by $\Pi_i$.
The strategy $\pi_i$ is said to be \emph{memoryless} if, for any state $s \in S$ and any two paths $\rho s, \rho's \in \Paths(G)$, we have $\pi_i(\rho s) = \pi_i(\rho' s)$.
The strategy is said to be \emph{deterministic} if, for every path $\rho \in \Paths(G)$, $\pi(\rho)$ is a singleton set.
A memoryless, deterministic strategy in $G$ is denoted as a map $\pi: S \rightarrow A_i$.
A pair of strategies $(\pi_1, \pi_2)$ is called a strategy profile.
Every strategy profile $(\pi_1, \pi_2)$ in $G$ determines a unique path denoted, with a slight abuse of notation, by $\Paths_G(s_0, \pi_1, \pi_2)$.
A strategy $\pi: S^+ \rightarrow A$ in $G$ is said to be \emph{proper} if, for every infinite path $s_0 s_1 \ldots \in \Paths^\infty(G)$, there exists an integer $n \geq 0$ such that $\pi_i(s_0 s_1 \ldots s_n)$ is undefined.
We assume all strategies considered in the paper to be proper.


\subsection{Specifying Temporal Goals}
\label{sec:ltlf}

The temporal goals of players in the game $G$ are specified formally using temporal logic formulas interpreted over finite traces \cite{de2013linear}. 

\begin{definition}
Given a   set of atomic propositions $AP$, a \ac{ltlf}  is produced by the following grammar:
\[
    \varphi \coloneqq  p \mid \neg \varphi \mid \varphi  \land \varphi  \mid \Next \varphi \mid \varphi  \until \varphi,
\]
made of atomic propositions $p \in AP$, the standard Boolean operators $\neg$ (negation) and $\land$ (conjunction), as well as temporal operators $\Next$ (``Next'') and $\until$ (``Until'').

\end{definition}
The dual of $\land$ is $\lor$ (disjunction), which is defined in the usual way using $\land$ and $\neg$, that is, $\varphi_1 \lor \varphi_2 := \neg(\neg \varphi_1 \wedge \varphi_2)$.
The temporal operators are used to specify properties of the system over sequences of time instants.
The formula $\Next \varphi$ indicates that $\varphi$ holds true at the next time instant. Formula $\varphi_1 \until \varphi_2$ means there is future time instant at which $\varphi_2$ holds and at all instant from now until that instant, $\varphi_1$ holds true.
From those temporal operators, two additional temporal operators $\Eventually$ (``Eventually'') and $\Always$ (``Always'') are defined.
The formula $\Eventually \varphi$ means $\varphi$ holds true at a future time instant, while $\Always \varphi$ means $\varphi$ holds true at the current instant and all future instants.
Formally, $\Eventually \varphi := \truev \until \varphi$ and $\Always \varphi := \neg \Eventually \neg \varphi$.
See~\cite{de2013linear} for formal semantics of \ac{ltlf}.

Every \ac{ltlf} formula over $AP$ defines a regular language over the alphabet $\Sigma = 2^{AP}$, denoted $\lang(\varphi)$.
Such a regular language can be specified by a finite automaton.

\begin{definition}
	A \ac{dfa} is a tuple $\calA = \langle Q, \Sigma, \delta, q_0, F \rangle$ where 
 	$Q$ is a finite state space.
 	$\Sigma$ is a finite alphabet. 
 	$\delta: Q \times \Sigma \rightarrow Q$ is a deterministic transition function.
    $q_0\in Q$ is an initial state, and $F\subseteq Q$ is a set of accepting (final) states.
\end{definition}


A transition from a state $q \in Q$ to a state $q' \in Q$ using input $\sigma \in \Sigma$ is denoted by $\delta(q, \sigma) = q'$. 
Slightly abusing the notation, we define the extended transition function $\delta: Q \times \Sigma^* \rightarrow Q$ as follows: $\delta(q, \sigma w) = \delta( \delta(q, \sigma), w )$ for each $w \in \Sigma^\ast$ and $\sigma \in \Sigma$, and $\delta(q, \epsilon) = q$ for each $q \in Q$, where $\epsilon$ is the empty string. The language of a \ac{dfa} $\calA$, denoted $\lang(\calA)$, consists of those words that induce a visit to an accepting state when input to the \ac{dfa}. Formally, $\lang(\calA) = \{ w \in \Sigma^* \mid \delta(q, w) \in F \}$. For every \ac{ltlf} formula $\varphi$ over $AP$, there exists a \ac{dfa} such that $\lang(A_{\varphi}) = \lang(\varphi)$ \cite{de2013linear}.

\subsection{Preference Modeling}

A preference model formally captures the notion of preferences in decision-making by representing the comparison between outcomes as a binary relation.
\begin{definition}
	\label{def:model_preference}
  	Given a countable set of outcomes $U$, a preference model over $U$ is a tuple $\langle U, \succeq \rangle$ in which $\succeq$ is \emph{preorder} on $U$, \ie, a reflexive and transitive binary relation on $U$.
\end{definition}

A preference model is defined as a binary relation on a countable and possibly infinite set $U$ \cite{bouyssou2013decision}.
A binary relation $\succeq$ on $U$ is a subset of $U \times U$.
The relation $\succeq$ is said to be \emph{reflexive} if and only if, for all $u \in U$, we have $(u, u) \in \succeq$.
It is said to be \emph{transitive} if and only if, for all $u_1, u_2, u_3 \in U$, $(u_1, u_2), (u_2, u_3) \in \succeq$ implies $(u_1, u_3) \in \succeq$.
The relation $\succeq$ is called a \emph{preorder} on $U$ if it is a reflexive and transitive binary relation on the set $U$.

The set of maximal elements in $U$ under a preorder $\succeq$ is the set $ \maximal(U, \succeq) = \{u \in U \mid \nexists u' \in U: u' \succeq u\}$. 
Similarly, the set of minimal elements in $U$ under a preorder $\succeq$ is the set $\mathsf{Min}(U, \succeq) = \{u \in U \mid \nexists u' \in U: u \succeq u'\}$.
For every non-empty $U$, the sets $\maximal(U, \succeq)$ and $\mathsf{Min}(U, \succeq)$ are always non-empty \cite{sen1997maximization}. 
The elements of the set $\maximal(U, \succeq)$ are called non-dominated or undominated elements of $U$.

We consider the preference language \prefltlf \cite{rahmani2024preference} to express preferences over temporal goals.

\begin{definition}
	\label{def:prefltlf} 
	A \prefltlf formula over a set of \ac{ltlf} formulas $\Phi$ is defined using the following grammar:
    \[
        \psi := \varphi_1  \weakpref \varphi_2 \mid \psi_1 \prefAnd \psi_2,
    \]
    where $\varphi_1$ and $\varphi_2$ are \ac{ltlf} formulas in $\Phi$, $\weakpref$ is a preference operator,
    and $\prefAnd$ is a generalized AND-operator.
    The formula $\varphi_1 \weakpref \varphi_2$ is called an \emph{atomic} preference formula.
\end{definition}

Each atomic \prefltlf~formula compares two \ac{ltlf} formulas, and accordingly, each \prefltlf formula specifies a collection of comparisons between \ac{ltlf} formulas.

We define additional preference operators to express strict preference ($\psi_1 \strictpref \psi_2$), indifference ($\psi_1 \indifferent \psi_2$), and incomparability ($\psi_1 \parallel \psi_2$).
The operators are understood as follows.
Let $\varphi_1, \varphi_2$ be two \ac{ltlf} formulas.
The formula $\varphi_1 \indifferent \varphi_2$ holds whenever $\varphi_1 \weakpref \varphi_2$ and $\varphi_2 \weakpref \varphi_1$ (\ie, $\varphi_1$ and $\varphi_2$ are \emph{indifferent});
$\varphi_1 \strictpref \varphi_2$ holds whenever $\varphi_1 \weakpref \varphi_2$ and $\varphi_2 \not\weakpref \varphi_1$ (\ie, $\varphi_1$ is \emph{strictly preferred} to $\varphi_2$ ); and
$\varphi_1 \parallel \varphi_2$ holds whenever $\varphi_1 \not\weakpref \varphi_2$ and $\varphi_2 \not\weakpref \varphi_1$ (\ie, $\varphi_1$ and $\varphi_2$ are \emph{incomparable}).


Every \prefltlf~formula defines a preorder on the set of words in $\Sigma^*$, which can be represented using a preference automaton.
A preference automaton is a computation model representing the relation defined by a \prefltlf~formula $\psi$ in \refDef{def:prefltlf}.

\begin{definition}
	\label{def:preference-automaton}
	A preference automaton for an alphabet $\Sigma$ is a tuple
	$$\pa= \langle Q, \Sigma, \delta, q_0, E \rangle,$$
	where
	\begin{inparaenum}[]
		\item $Q$ is a finite set of states.
		\item $\Sigma$ is the alphabet.
		\item $\delta: Q \times \Sigma \rightarrow Q$ is a deterministic transition function.
		\item $q_0 \in Q$ is the initial state.
		\item $E \subseteq Q \times Q$ is a preorder on $Q$.
	\end{inparaenum}
\end{definition}

Note that Def.~\ref{def:preference-automaton} augments the semi-automaton $\langle Q, \Sigma, \delta, q_0 \rangle$ with the preference relation $E$, instead of a set of accepting (final) states as is typical with a \ac{dfa}.
We write $q \succeq_E q'$ to denote that state $q$ is weakly preferred to $q'$ under preorder $E$.

The preference automaton encodes a preference relation $\succeq$ on $\Sigma^*=(2^{AP})^*$ as follows.
Consider two words $w, w' \in \Sigma^*$.
Let $q, q' \in Q$ be the two states such that $q = \delta(q_0, w)$ and $q' = \delta(q_0, w')$.
There are four cases:
\begin{inparaenum}[(a)]
	\item If $(q, q') \in E$ and $(q', q) \notin E$, then $w \succ w'$;
	\item If $(q, q') \notin E$ and $(q', q) \in E$, then $w' \succ w$;
	\item If $(q, q') \in E$ and $(q', q) \in E$, then $w \sim w'$;
	\item If $(q, q') \notin E$ and $(q', q) \notin E$, then $w \parallel w'$.
\end{inparaenum}

The procedure to construct a preference automaton from a \prefltlf~formula is enlisted in \cite{rahmani2023probabilistic}\footnote{A tool to translate a \prefltlf~formula to a preference automaton is available at \url{https://akulkarni.me/prefltlf2pdfa.html}.}.

\section{Problem Formulation}
\label{sec:problem-formulation}

A deterministic game with incomplete preference is a deterministic game $G$ (see \refDef{def:dtptb-game}) in which players aim to maximally satisfy their (possibly incomplete) preferences expressed as \prefltlf~formulas $\psi_1, \psi_2$, respectively, over a set of \ac{ltlf} formulas $\Phi$.
We denote a deterministic game with incomplete preference as a tuple, $\langle G, \Phi, \psi_1, \psi_2 \rangle$.

%
%
%
%
%

A Nash equilibrium in $\calG$ is a strategy profile such that no player can achieve a strictly preferred outcome by unilaterally changing their strategy. 
The following definition adapts the standard definition of Nash equilibrium \cite{bade2005nash} to a deterministic game with preferences.




\begin{definition}
	\label{def:nash}
	Let $\succeq_i$ be the preference relation induced by $\psi_i$ on the set of finite words $\Sigma^* = (2^{AP})^*$. 
	A strategy profile $(\pi_1^*, \pi_2^*)$ is a \ak{pure} Nash equilibrium in $\calG$ if and only if the following conditions hold.
	\begin{enumerate}[i)]
		\item There does not exist a strategy $\pi_1 \in \Pi_1$ such that $\Paths_\calG(s_0, \pi_1, \pi_2^*) \succeq_1 \Paths_\calG(s_0, \pi_1^*, \pi_2^*)$.
		
		\item There does not exist a strategy $\pi_2 \in \Pi_2$ such that $\Paths_\calG(s_0, \pi_1^*, \pi_2) \succeq_2 \Paths_\calG(s_0, \pi_1^*, \pi_2^*)$.
	\end{enumerate}
\end{definition}

We now state our problem statement.

\begin{problem}
	\label{prob:synthesis-with-preferences}
	Given a deterministic game with incomplete preferences $\calG = \langle G, \Phi, \psi_1, \psi_2 \rangle$, determine the set of Nash equilibria in $\calG$.
\end{problem}

\section{Main Results}
\label{sec:main-results}

In this section, we characterize the set of Nash equilibrium in games with incomplete preferences.
A key insight we obtain is that Nash equilibria can be characterized based on the well-known solution concepts of sure winning and Pareto efficiency in game theory, depending on different alignment between player preferences and whether the players are cooperative or agnostic.  
Hence, we study the characterization of Nash equilibria for each case separately.


We follow an automata-theoretic approach to characterize the Nash equilibria since it enables transforming a preference relation over \ac{ltlf} objectives to a preference relation over states of the product game defined below.
%

%
%
%

\begin{definition}
	\label{def:product-game}
	Given a game $\calG = (G, \Phi, \psi_1, \psi_2)$ and the preference automata $\pa_1 = (Q, \Sigma, \delta, q_{0}, E_1)$  and $\pa_2 = (Q, \Sigma, \delta, q_{0}, E_2)$ induced by players' preferences $\psi_1, \psi_2$, respectively, the product game is a tuple,
	\begin{align*}
		H = \langle V, A, \Delta, v_0, \calE_1, \calE_2\rangle,
	\end{align*}
	where
	\begin{inparaenum}[]
		\item $V = S \times Q$ is the set of states.  
		$V_1 = S_1 \times Q$ are P1 states and $V_2 = S_2 \times Q$ are P2 states.
		\item $A = A_1 \cup A_2$ is the set of actions.
		\item $\Delta: V \times A \rightarrow V$ is a deterministic transition function.
		Given two states $v = (s, q), v' = (s', q') \in V$ and an action $a \in A$, we have $\Delta(v, a) = v'$ if and only if $s' = T(s, a)$ and $q' = \delta(q_{0}, L(s'))$.
		\item $v_0 = (s_0, \delta(q_{0}, L(s_0)))$ is the initial state.
		\item $\calE_1, \calE_2$ are preorders on $V$ such that $(s, q) \succeq_{\calE_1} (s', q')$ if and only if $q \succeq_{E_1} q'$ and  $(s, q) \succeq_{\calE_2} (s', q')$ if and only if $q_2 \succeq_{E_2} q_2'$.
	\end{inparaenum}
\end{definition}

Notice that the state space of $H$ is defined as $S \times Q$ instead of the usual $S \times Q \times Q$ used when defining the product of a game with two automata \cite{rahmani2023probabilistic}. 
This compact state representation is possible because the two automata, $\pa_1$ and $\pa_2$, share the same semi-automaton $\langle Q, \Sigma, \delta, q_0 \rangle$, which is defined by the shared set of alternatives $\Phi$ for both players in $\calG$. 
The components $A, \Delta, v_0$ of the product game are defined according to standard game and DFA product construction \cite{baier2008principles}. The preorder relations $\calE_1$ and $\calE_2$ on $V$ are lifted from the relations $E_1$ and $E_2$. 
For instance, P1 strictly prefers a state $(s, q)$ to $(s', q')$ if $q$ is strictly preferred to $q'$ under P1's preference relation $E_1$ in its automaton $\pa_1$.

The following proposition establishes that a product game transforms the preference over satisfying \ac{ltlf} objectives in $G$ to a preference relation over states in $H$.
First, we define a notation. 
Every path $\rho = s_0 s_1 \ldots s_n$ in $G$ induces a unique path $\varrho = v_0 v_1 \ldots v_n$ in $H$ where, for all $j = 0, \ldots, n$, $v_j = (s_i, q_j)$ and $q_j = \delta(q_0, L(s_0 s_1 \ldots s_j))$. 
We call the path $\varrho$ as the \emph{trace} of $\rho$ in $H$.

\begin{proposition}
	\label{prop:pref(paths in G)-to-pref(states-in-H)}
	Given any finite paths $\rho, \rho'$ in $G$, let $\varrho, \varrho'$ be their traces in $H$.
	Then, for $i = 1, 2$, $L(\rho) \succeq_i L(\rho')$ if and only if $\last(\varrho) \succeq_{\calE_i} \last(\varrho')$.
\end{proposition}
\begin{proof}
	Suppose that $L(\rho) \succeq_i L(\rho')$ holds. 
	By definition of the preference automaton, $L(\rho) \succeq_i L(\rho')$ implies $q \succeq_{E_i} q'$, where $q = \delta(q_0, L(\rho))$ and $q' = \delta(q_0, L(\rho'))$. 
	Given that the preference relation $\calE_i$ in the product game is a lifting of the relation $E_i$, it follows that $(s, q) \succeq_{\calE_i} (s', q')$, where $(s, q) = \last(\varrho)$ and $(s', q') = \last(\varrho')$. 
	Since the converse of each aforementioned statement is true by definition, the proposition is thus established.
\end{proof}


\refProp{prop:pref(paths in G)-to-pref(states-in-H)} implicitly defines a preference relation on set of strategy profiles in $H$, because every strategy profile defines a unique path in $H$.
Given two strategy profiles, $(\pi_1, \pi_2)$ and $(\pi_1', \pi_2')$, $(\pi_1, \pi_2)$ is weakly preferred to $(\pi_1', \pi_2')$ for player-$i$ if and only if the last state visited by the path $\Paths_H(v_0, \pi_1, \pi_2)$ is strictly preferred to that visited by $\Paths_H(v_0, \pi_1', \pi_2')$ under $\calE_i$. 
We say $(\pi_1, \pi_2)$ dominates $(\pi_1', \pi_2')$ when $(\pi_1, \pi_2)$ is strictly preferred to $(\pi_1', \pi_2')$.

%
%
%
%
%
%
%

Next, we discuss characterization of Nash equilibrium strategies for the three cases of preference alignment.

\subsection{Fully Aligned Preferences}
The player preferences are said to be fully aligned when $\calE_1 = \calE_2$. 
In this case, any unilateral change in strategy that benefits one player also benefits the other player.
Consequently, every Nash equilibrium must satisfy the maximal outcome that is realizable under full cooperation between the players. 


Let $\Pi_1, \Pi_2$ be the set of all P1 and P2 strategies in $H$. 
We use $\reach_H(v)$ to denote the subset of $V$ that is reachable from a state $v \in V$.
\begin{align*}
	\reach_H(v) = \{v'  &\in V \mid \exists \pi_1 \in \Pi_1, \exists \pi_2 \in \Pi_2: \\ 
	&\rho = \Paths_H(v_0, \pi_1, \pi_2) \text{ and } \last(\rho) = v'\}
\end{align*}
Given a preorder $\calE$, a state $v \in V$, and the set $U = \reach_H(v)$, every state in $\maximal(U, \calE)$ is called the maximal reachable state from $v$ under $\calE$.

\begin{theorem}
	\label{thm:aligned:nash}
	A strategy profile $(\pi_1, \pi_2)$ is a Nash equilibrium in $H$ if and only if the last state of the path $\Paths_H(v_0, \pi_1, \pi_2)$ is a maximal reachable state from $v_0$ under $\calE = \calE_1 = \calE_2$.
\end{theorem}
\begin{proof}
	($\implies$).
	We will prove that $(\pi_1, \pi_2)$ is a Nash equilibrium if $\rho = \Paths_H(v_0, \pi_1, \pi_2)$ visits a maximal reachable state from $v_0$.
	For this, two conditions must hold. 
	First, there does not exist a P1 strategy $\pi_1' \in \Pi_1$ such that $(\pi_1', \pi_2)$ strictly dominates $(\pi_1, \pi_2)$. 
	Second, there does not exist a P2 strategy $\pi_2' \in \Pi_2$ such that $(\pi_1, \pi_2')$ strictly dominates $(\pi_1, \pi_2)$.
	Consider the first case.
	For $(\pi_1', \pi_2)$ to strictly dominate $(\pi_1, \pi_2)$, the last state of the path $\Paths_H(v_0, \pi_1', \pi_2)$ must be strictly preferred to the last state visited by $\rho$.
	However, this is not possible because the last state of $\rho$ is a maximal state in $\maximal(\reach_H(v_0), \calE)$.
	Therefore, the first condition must be true.
	For a similar reason, the second condition also holds.

	($\impliedby$).
	If the path $\rho$ terminates at a maximal reachable state from $v_0$ under $\calE$, then $(\pi_1, \pi_2)$ is an undominated strategy profile since there is no state that is reachable from $v_0$ and strictly preferred to $\last(\rho)$.
	Therefore, $(\pi_1, \pi_2)$ must be a Nash equilibrium.
\end{proof}

From \refThm{thm:aligned:nash}, it follows that the set of Nash equilibria in any game $H$ is non-empty because $\reach(v_0)$ contains at least $v_0$ itself.
Since there can be more than one maximal states in $V$ under $\calE$, there may exist multiple Nash equilibria.

\begin{corollary}
	In every game $H$, there exists at least one Nash equilibrium. There exist games with more than one Nash equilibria.
\end{corollary}

\subsection{Completely Opposite Preferences}

The player preferences are said to be completely opposite whenever $v \succeq_{\calE_1} v'$ implies $v' \succeq_{\calE_2} v$.
That is, if P1 strictly prefers a state $v$ to $v'$ then P2 strictly prefers $v'$ to $v$.

\begin{proposition}
	When player preferences are completely opposite, any two states in $V$ that are incomparable under $\calE_1$ are also incomparable under $\calE_2$
\end{proposition}
\begin{proof}
	Using a set-theoretic representation of preorders, we write $(v, v') \in \calE_1$ to represent $v \succeq_{\calE_1} v'$. 
	%
	Let $v, v' \in V$ be two states such that $v \parallel_{\calE_1} v'$. 
	Since $v \parallel_{\calE_1} v'$ holds if and only if $v \not\succeq_{\calE_1} v'$ and $v' \not\succeq_{\calE_1} v$, neither $(v, v')$ nor $(v', v)$ is an element of $\calE_1$. 
	By definition of completely opposite preferences, it follows that neither $(v, v')$ nor $(v', v)$ belong to $\calE_2$. 
	That is, $v$ and $v'$ are incomparable under $\calE_2$.
\end{proof} 

Unlike games with fully aligned preferences, where Nash equilibria can be found by computing the maximal reachable states, games with completely opposite preferences require a different approach.
In these games, a better outcome for one player means a worse outcome for the other. 
As a result, each player aims to synthesize a strategy that achieves the best possible outcome for themselves while preventing their opponent from improving their outcome.

We introduce the concept of a \emph{maximal sure winning strategy} for a player, which ensures achieving the best possible outcome against any strategy the opponent might employ. 
A maximal sure winning strategy extends the concept of sure winning in games on graphs \cite{de2007concurrent}.
Intuitively, it enables the player to react to each potential move by the opponent by reasoning about multiple outcomes simultaneously, choosing an action that achieves the best possible outcome for that player.


Given that a player’s best achievable outcome might not be their most-preferred,
we need a way to compare the quality of a given outcome to the most-preferred outcome for that player.
For this purpose, we introduce the notion of a rank.



\begin{definition}
	\label{def:rank}
	Given a preorder $\calE$ on $V$, let $Z_0 = \maximal(V, \calE)$ and $Z_k$ be defined inductively as $Z_{k+1} = \maximal(V \setminus \bigcup\limits_{j=0}^{k} Z_{j}, \calE)$. Then, given a state $v \in V$, the smallest integer $k \geq 0$ such that $v \in Z_k$ is called the rank of state $v$, denoted by $\rank_\calE(v) = k$.
\end{definition}

\refDef{def:rank} assigns a unique, finite rank to every state in $V$ given a preorder $\calE$. 
Since the set $\maximal(U, \calE)$ is non-empty for any non-empty subset $U \subseteq V$, the inductive assignment of ranks terminates only when the subset $V \setminus \bigcup\limits_{j=0}^{k} Z_{j}$ is empty, \ie, when a rank has been assigned to all states in $V$.
Additionally, the sets $Z_0, Z_1 \ldots$ are mutually exclusive and exhaustive subsets of $V$. 
Therefore, every state in $V$ has a unique rank under a given preorder.


\begin{proposition}
	\label{prop:rank-comparison}
	The following statements hold for any two states $v, v' \in V$,
	\begin{enumerate}
		\item If $\rank_\calE(v) = \rank_\calE(v')$ then either $v \sim_{\calE} v'$ or $v \parallel_{\calE} v'$.
		\item If $\rank_\calE(v) > \rank_\calE(v')$ then $v \not\succeq_{\calE} v'$.
		\item If $v \succ_{\calE} v'$ then $\rank_\calE(v) < \rank_\calE(v')$.
	\end{enumerate}
\end{proposition}
\begin{proof}
	(1) We will first show that when $\rank_\calE(v) = \rank_\calE(v') = k$, neither $v \succ_\calE v'$ nor $v' \succ_\calE v$ can be true. 
	Suppose that $v \succ_{\calE} v'$ is true. 
	Then, by \refDef{def:rank}, both states $v$ and $v'$ must be elements of $Z_k$, which means that $v$ and $v'$ must be elements of the set $\maximal(Y, \calE)$ where $Y = V \setminus (Z_0 \cup Z_1 \cup \ldots \cup Z_{k-1})$. 
	But $v$ and $v'$ cannot both be maximal elements of $Y$ because $v \succ_{\calE} v'$, which contradicts our supposition. 
	A similar argument can be used to establish that $v' \succ_{\calE} v$. 
	If neither $v \succ_{\calE} v'$ nor $v' \succ_{\calE} v$ is true, then it must be the case that either $v \sim_{\calE} v'$ or $v \parallel_{\calE} v'$.

	(2) Let $\rank_\calE(v') = k$. 
	If $\rank_\calE(v) > \rank_\calE(v')$ then, by \refDef{def:rank}, $v$ and $v'$ are both included in the set $Y = V \setminus (Z_0 \cup Z_1 \cup \ldots \cup Z_{k-1})$. 
	If $v \succeq_\calE v'$, then by definition it must be included in $Z_k = \maximal(Y, \calE)$. 
	Since this is not the case, the statement $v \not\succeq_{\calE} v'$ must be true.
	
	(3) The statement follows by a similar argument as (2).
\end{proof}

\refProp{prop:rank-comparison} establishes the relationship between preference between two states and the comparison of their ranks. 
First, it states that any two states states with equal ranks are either indifferent or incomparable to each other under the given preorder. 
Second, it states that a state with a higher rank is no better than one with a lower rank.
Lastly, it states that a state that is strictly preferred to another has a strictly smaller rank than the other.

\begin{remark}
	We note that the converse of statements in \refProp{prop:rank-comparison} do not necessarily hold. 
	This is mainly because incomplete preferences can lead to situations where outcomes are incomparable. 
	That is, none of the following statements is valid: 
	($1'$) If $v \parallel_{\calE} v'$ then $\rank_\calE(v) = \rank_\calE(v')$. 
	($2'$) If $v \not\succeq_{\calE} v'$ then $\rank_\calE(v) > \rank_\calE(v')$.
	($3'$) If $\rank_\calE(v) < \rank_\calE(v')$ then $v \succ_{\calE} v'$. 
	
	As an counterexample, consider a game with five states $\{v_1, \ldots, v_5\}$ where the state $v_1$ is strictly preferred to $v_2$, $v_3$ is strictly preferred to $v_4$, and $v_5$ is strictly preferred to $v_4$. 
	Following \refDef{def:rank}, the states $v_1, v_3, v_5$ have rank $0$ and the states $v_2, v_4$ have rank $1$. 
	Statement (a) is invalid because $v_2$ and $v_3$ are incomparable but they have different ranks. 
	Statement (c) is invalid because the rank of $v_3$ is smaller than that of $v_3$ but $v_2 \succ_{\calE} v_3$ does not hold since are $v_2$ and $v_3$ incomparable.
	To see the invalidity of statement (b), first note that $v \not\succeq_{\calE} v'$ holds when either $v' \succ_{\calE} v$ or $v \parallel_{\calE} v'$. 
	Now, observe that states $v_3$ and $v_5$ are incomparable but they have the same ranks.  
\end{remark}

In every two-player game with preferences, each player assigns a rank to every state. 
Let $\rank_1(v)$ and $\rank_2(v)$ denote the rank of state $v$ under $\calE_1$ and $\calE_2$, respectively. 
Additionally, let $k_1^{\max} \triangleq \max \{\rank_1(v) \mid v \in V\}$ and $k_2^{\max} \triangleq \max \{\rank_2(v) \mid v \in V\}$ denote the largest rank assigned to any state in $V$ under $\calE_1$ and $\calE_2$.


 
The following proposition shows that the game $H$, where player preferences are completely opposite, can be represented as a constant-sum game in which each player's payoff is the rank they assign to the last state.

\begin{proposition}
	\label{prop:opposite.constant-sum}
	When player preferences are completely opposite, for every $v \in V$, the sum of ranks assigned to $v$ under $\calE_1$ and $\calE_2$ is constant, \ie, $\rank_1(v) + \rank_2(v) = k_1^{\max} = k_2^{\max}$.
\end{proposition}
\begin{proof}
	We will show that for any set $U \subseteq V$, a maximal state in $U$ under $\calE_1$ is the minimal element in $U$ under $\calE_2$ (recall that a state $v \in V$ is minimal under $\calE_1$ if there is no state $u \in V$ such that $v \succeq_{\calE_1} u$).
	
	First, we note that the minimal states in $V$ under $\calE_1$ are all included in the set $Z_{k_1^{\max}}$.
	This is because $k_1^{\max}$ is the maximum rank assigned to any state in $V$ and if there were a state $u \in V$ which was minimal but not included in $Z_{k_1^{\max}}$, then it must have a rank greater than ${k_1^{\max}}$, by \refProp{prop:rank-comparison}.
	
	Now, consider a state $v \in Z_{k_1^{\max}}$.
	We will show that $v$ is a maximal state under $\calE_2$, \ie, $\rank_2(v) = 0$.
	For this, we observe that every state $u \in Z_j$ for any $j < {k_1^{\max}}$ satisfies $u \succeq_{\calE_1} v$ or $u \parallel_{\calE_1} v$.
	Thus, under the opposite preference relation $\calE_2$ it must satisfy $v \succeq_{\calE_2} u$ or $v \parallel_{\calE_2} u$.
	Since $v$ was a minimal element in $V$ under $\calE_2$, there is no $u$ such that $u \succeq_{\calE_2} v$.
	In other words, $v$ is a maximal element in $V$ under $\calE_2$.
	By definition, $\rank_2(v) = 0$.
	
	It follows that every $v$ such that $\rank_1(v) = {k_1^{\max}}$ has a rank $0$ under $\calE_2$.
	For $j = 0, 1, \ldots$, let $Y_j$, , denote the set of states with rank $j$ under $\calE_2$.
	Using a similar argument, the minimal elements of $\bigcup\limits_{j=0}^k Z_j$ are the maximal elements of the set $V \setminus \bigcup\limits_{j=0}^k Y_j$.
	Therefore, every state in $Y_j$ is a state with rank ${k_1^{\max}} - j$ under $\calE_2$.
	
	It follows that $\rank_1(v) + \rank_2(v) = {k_1^{\max}}$.
\end{proof}

A key insight from \refProp{prop:opposite.constant-sum} is that P1 must play a strategy aimed at minimizing the rank under $\calE_1$ to achieve the best possible outcome. 
With this insight, we formalize the concept of maximal sure winning strategy.





First, we define a notation. 
Given a strategy profile $(\pi_1, \pi_2)$, we write $\rank_i(s_0, \pi_1, \pi_2)$ to denote the rank of the last state of the path $\Paths_H(s_0, \pi_1, \pi_2)$ under $\calE_i$.
We denote the maximum rank achievable under $\calE_i$ when P1 follows strategy $\pi_1$ by $\maxRank_i(\pi_1) \coloneqq \max \{\rank_i(\pi_1, \pi_2) \mid \pi_2 \in \Pi_2\}$.
The minimum rank achievable for P1 under $\calE_i$ by following $\pi_1$ is defined similarly.

\begin{definition}
	\label{def:maximal-swin}
	A strategy $\pi_1 \in \Pi_1$ in $H$ is said to be \emph{maximal sure winning} for P1 if there does not exist a P1 strategy $\pi_1' \in \Pi_1$ such that $\maxRank_1(\pi_1') < \maxRank_1(\pi_1)$.
\end{definition}

The maximal sure winning strategy for P2 is defined analogously.


Intuitively, \refDef{def:maximal-swin} indicates that a P1's maximal sure winning strategy minimizes the maximum rank achieved under any P1 strategy, regardless of P2's strategy.
However, this does not mean that P1's maximal sure winning strategy is undominated in $\Pi_1$.
Specifically, if $\pi_1$ is a P1's maximal sure winning strategy then, given a P2 strategy $\pi_2$, there may exist another P1 strategy $\pi_1'$ such that $\rank_1(\pi_1', \pi_2) < \rank_1(\pi_1, \pi_2) \leq \min\limits_{\pi_1 \in \Pi_1} \maxRank_1(\pi_1)$.

%

But would a rational P2 play a strategy $\pi_2$ such that $\rank_1(\pi_1, \pi_2) < \maxRank_1(\pi_1)$, where $\pi_1$ is a P1's maximal sure winning strategy? 
The following proposition answers this question.

\begin{proposition}
	\label{prop:opposite.p2-max-swin}
	Let $\pi_1$ be a maximal sure winning strategy.
	If $\maxRank_1(\pi_1) = k$, then every maximal sure winning strategy $\pi_2$ of P2 satisfies $\maxRank_2(\pi_2) = k_2^{\max} - k$.
\end{proposition}
\begin{proof}
	To answer this question, we observe that whenever $\maxRank_1(\pi_1) = k$, there exists a P2 strategy $\pi_2^*$ such that $\rank_1(\pi_1, \pi_2^*) = k$. 
	By \refProp{prop:opposite.constant-sum}, P2 must follow $\pi_2^*$ since $\rank_1(\pi_1, \pi_2^*)$ is the maximum possible rank under $\calE_1$, or equivalently, the smallest possible rank under $\calE_2$ possible when P1 plays its maximal sure winning strategy. 
\end{proof}

\refProp{prop:opposite.p2-max-swin} establishes that when P1 plays its maximal sure winning strategy, a rational P2 must also play its maximal sure winning strategy. 
And when both players play their maximal sure winning strategies, neither P1 cannot achieve a smaller rank than $\maxRank_1(\pi_1)$ nor P2 can achieve a smaller rank than $\maxRank_2(\pi_2)$.

The following theorem encodes a procedure to compute maximal sure winning strategy for a player. 




\begin{theorem}
	\label{thm:max-swin}
	Let $\pi_1$ be a sure winning strategy to visit the set $Y_k = \{v \in V \mid \rank(v) \leq k\}$ where $k$ is the smallest integer such that $v_0 \in \swin_1(Y_k)$.
	Then, $\pi_1$ is a maximal sure winning strategy for P1.
\end{theorem}
\begin{proof}
	By contradiction. 
	Suppose that $\pi_1$ is not a maximal sure winning strategy for P1. 
	According to \refDef{def:maximal-swin}, there must exist another strategy $\pi_1'$ for P1 such that $\maxRank_1(\pi_1') < \maxRank_1(\pi_1)$. 
	Let $\maxRank_1(\pi_1') = m$ and $\maxRank_1(\pi_1) = k$. 
	Clearly, $m < k$.
	
	Now, we observe two facts. 
	First, when P1 follows $\pi_1'$, the game reaches a terminal state within the set $Y = \{v \in V \mid \exists \pi_2 \in \Pi_2: v \text{ is the last state of } \Paths_H(v_0, \pi_1', \pi_2)\}$. 
	Second, $Y$ is a subset of $Y_m = \{v \in V \mid \rank_1(v) \leq m\}$ because $\maxRank_1(\pi_1') = m$.
	
	Consequently, when P1 follows $\pi_1'$, it is guaranteed to reach a state in $Y_m$ regardless of the strategy employed by P2. This means $\pi_1'$ is a sure winning strategy for P1 to reach $Y_m$, which contradicts the assumption that $k$ is the smallest integer such that $v_0 \in \swin_1(Y_k)$.
	Therefore, $\pi_1$ must be a maximal sure winning strategy for P1.
\end{proof}
The following corollary to \refThm{thm:max-swin} establishes the existence of a maximal sure winning strategy for both players in every game.

\begin{corollary}
	In every game $H$, both players have a maximal sure winning strategy.
\end{corollary}
\begin{proof}
	For a player, say P1, to have a maximal sure winning strategy, there must exist some $k \geq 0$ such that $v_0 \in \swin_1(Y_k)$, where $Y_k = \{v \in V \mid \rank_1(v) \leq k\}$. 
	This condition is always satisfied for $k = \rank_1(v_0)$. 
	In this case, we have $v_0 \in Y_k$. 
	The statement follows from the fact that, for all sure winning regions, $v_0 \in Y_k$ implies that $v \in \swin_1(Y_k)$ \cite{de2007concurrent}.
\end{proof}


\refThm{thm:max-swin} yields a procedure to compute the maximal sure winning region for a player. 
The procedure iteratively computes the sure winning regions $Y_0, Y_1, \ldots, Y_j$ until $v_0$ is included in $\swin_1(Y_j)$, for $j = 0, \ldots, k_1^{\max}$. 
Assuming $k$ is the smallest integer for which $v_0 \in Y_k$, \refThm{thm:max-swin} states that the sure winning strategy to visit $Y_k$ is the maximal sure winning strategy for the player.


%

Furthermore, the procedure to compute the maximal sure winning strategy scales linearly with the size of $H$ and the maximum rank of any state in the game, $k_1^{\max}$. 
This is because the procedure to compute $\swin_1(Y_j)$, which scales linearly with size of $H$ \cite{zielonka1998infinite}, is called at most $k_1^{\max}$-many times when solving for the maximal sure winning strategy.

 



Finally, we show that every Nash equilibrium in $H$ consists of P1 and P2's maximal sure winning strategies.

%


\begin{theorem}
	\label{thm:opposite.nash-equilibrium}
	Every strategy profile $(\pi_1, \pi_2)$ such that $\pi_1$ and $\pi_2$ are P1 and P2's maximal sure winning strategies is a Nash equilibrium in $H$.
\end{theorem}
\begin{proof}
	($\implies$). By contradiction. 
	Let $(\pi_1, \pi_2)$ be a Nash equilibrium in $H$. 
	Without loss of generality, suppose that $\pi_1$ is not a maximal sure winning strategy. 
	Since $\pi_1$ is not a maximal sure winning, there must exist a P1 strategy $\pi_1'$ such that $\maxRank_1(\pi_1') < \maxRank_1(\pi_1)$. 
	But this implies that $\rank_1(\pi_1', \pi_2) < \rank_1(\pi_1, \pi_2)$. 
	Since this violates the condition for $(\pi_1, \pi_2)$ to be a Nash equilibrium (see \refDef{def:nash}), it must be the case that $\pi_1$ is a maximal sure winning strategy for P1.	
	Using a similar argument, $\pi_2$ also must be a maximal sure winning strategy for P2.
	
	($\impliedby$). Let $\pi_1$ and $\pi_2$ be a maximal sure winning strategies of P1 and P2 in $H$.  
	Then, by \refProp{prop:opposite.constant-sum} and \refProp{prop:opposite.p2-max-swin}, we have that $\rank_1(\pi_1, \pi_2) = \maxRank_1(\pi_1) = k_2^{\max} - \maxRank_2(\pi_2)$.
	That is, there is no P1 strategy $\pi_1'$ such that $\rank_1(\pi_1', \pi_2) < \rank_1(\pi_1, \pi_2)$ and there is no P2 strategy $\pi_2'$ such that $\rank_2(\pi_1, \pi_2') < \rank_2(\pi_1, \pi_2)$. 
	Since this condition satisfies the two conditions in \refDef{def:nash}, $(\pi_1, \pi_2)$ is a Nash equilibrium in $H$.
\end{proof}
	

A key insight from \refThm{thm:opposite.nash-equilibrium} is that even though \refDef{def:maximal-swin} employs a weaker form of undominance to define maximal sure winning concept, it is sufficient to characterize the Nash equilibria when player preferences are completely opposite.


\subsection{Partial Aligned Preferences} 

The player preferences are said to be partially aligned if they are neither fully aligned nor completely opposite.
In this scenario, players might have an incentive to cooperate if their preferences align. 
However, they might also be motivated to compete if a better outcome for one player results in a worse outcome for the other. 
Therefore, to characterize the Nash equilibrium, it is important to determine when a player needs cooperation and when a player has an incentive to cooperate.

We say that a player needs cooperation if the best outcome achievable through cooperation is strictly preferred to the best outcome it can guarantee without cooperation. 

\begin{definition}
\label{def:need-of-cooperation}
	Let $\pi_i$ be a maximal sure winning strategy of player-$i$.
	We say a player-$i$ needs cooperation by other player if the following equation holds, 
	\begin{align}
		\maxRank_i(\pi_i) > \min \{\rank_i(\pi_1, \pi_2) \mid \pi_1 \in \Pi_1, \pi_2 \in \Pi_2\}
	\end{align}
\end{definition}

Even if one player needs cooperation, the other player may or may not have an incentive to cooperate.
Drawing from \cite{morschheuser2017games}, we identify two motivations for a player to cooperate. 
Instrumental cooperation refers to the case when the cooperating player benefits from the collaboration, achieving a better outcome than they would without it. 
Attitudinal cooperation, on the other hand, describes a scenario where the cooperating player gains no direct benefit but chooses to cooperate based on their altruistic tendencies. 
It is important to note that when a player engages in attitudinal cooperation, they do so only if restricting their strategy does not worsen their own outcome.

We now characterize the set of Nash equilibria based on the number of players who need cooperation. 
Without loss of generality, we assume that the state with the smallest rank reachable from the initial state in $H$ has a rank of $0$ for both players, \ie, $\min \{\rank_i(\pi_1, \pi_2) \mid \pi_1 \in \Pi_1, \pi_2 \in \Pi_2\} = 0$ for $i = 1, 2$.

\textbf{No player needs cooperation.}
In this case, both players have a maximal sure winning strategy that achieves a rank $0$ outcome regardless of the strategy used by their opponent. 
Hence, we have the following result.

\begin{theorem}
	\label{thm:partial.no-cooperation}
	Let $Y = \{v \in V \mid \rank_1(v) = \rank_2(v) = 0\}$.
	A strategy profile $(\pi_1, \pi_2)$ is a Nash equilibrium in $H$ if and only if the last state of the path $\Paths_H(v_0, \pi_1, \pi_2)$ is an element of the set $Y$.
\end{theorem}
\begin{proof}
	The proof has two parts.  
	First, we show that any strategy profile that induces a path that terminates at a state with rank greater than $0$ cannot be a Nash equilibrium. 
	Then, we observe that a strategy profile $(\pi_1, \pi_2)$ for which the path $\Paths_H(v_0, \pi_1, \pi_2)$ terminates at a state in $Y$ is a Nash equilibrium.
	
	Consider a strategy profile $(\pi_1', \pi_2')$ such that the last state of the path $\Paths_H(v_0, \pi_1', \pi_2')$ has a rank greater than $0$. 
	Clearly, $\rank_1(\pi_1, \pi_2) < \rank_1(\pi_1', \pi_2')$ because the maximal sure winning strategy $\pi_1$ terminates at a state with rank $0$.
	Therefore, $(\pi_1', \pi_2')$ cannot be a Nash equilibrium. 
	
	The second statement is true because there is no strategy profile that can achieve a better rank than $0$. 
	This concludes the proof. 
\end{proof}
%

\textbf{Only one player needs cooperation.}
In this case, the player who does not need cooperation has a maximal sure winning strategy that ensures a rank $0$ outcome for the player.
The set of Nash equilibria is then determined by the attitude of the player who does not need cooperation.

Without loss of generality, let P1 be the player who needs cooperation.
\begin{theorem}
	\label{thm:attitude.cooperative}
	When P2 is attitudinally cooperative, a strategy profile $(\pi_1, \pi_2)$ is a Nash equilibrium if and only if $\pi_2$ is a maximal sure winning strategy for P2 and the last state of the path $\Paths_H(v_0, \pi_1, \pi_2)$ is an element of the set $Y$, where $Y = \arg\min\limits_{v \in V} \{\rank_1(v) \mid \rank_2(v) = 0\}$.
\end{theorem}
\begin{proof}
	($\implies$). Let $(\pi_1, \pi_2)$ be a Nash equilibrium. 
	The reason why $\pi_2$ must be a maximal sure winning strategy for P2 can be established using a similar argument to the proof of \refThm{thm:opposite.nash-equilibrium}. 
	Recall that an attitudinally cooperative P2 means that P1 achieves a maximal outcome with the set $W = \{v \in V \mid \rank_1(v) = 0\}$. 
	By \refProp{prop:rank-comparison}, the maximal elements in $W$ are those with the smallest rank, which is given by $Y$. 
	
	($\impliedby$). Let $\pi_2$ be maximal sure winning strategy of P2 and $\pi_1$ be a strategy such that $(\pi_1, \pi_2)$ induce a path that terminates in $Y$.
	Clearly, P2 has no strategy that achieves a better outcome than $\pi_2$ because $\pi_2$ ensures a rank $0$ outcome to P2.	
	Also, P1 has no strategy that achieves a better outcome than $\pi_1$ because $Y$ contains the maximal elements in $Z$.
	By \refDef{def:nash}, $(\pi_1, \pi_2)$ must be a Nash equilibrium.
\end{proof}

In words, when P2 has a cooperative attitude, the Nash equilibrium yields an outcome from the states with a rank of $0$ under $\calE_2$ that have least possible rank under $\calE_1$.

When P2 is not attitudinally cooperative, P1 must develop a strategy that achieves the best possible rank, assuming that P2 may use any of its maximal sure winning strategies. 
To this end, we define a sub-game of game $H$ as follows: $\widehat{H} = \langle V, A, \widehat{\Delta}, v_0, \calE_1, \calE_2\rangle,$ where $V, A, v_0, \calE_1, \calE_2$ have the same meanings as \refDef{def:product-game} and the transition function $\widehat{\Delta}$ is defined as follows: $\widehat{\Delta}(v, a) = \Delta(v, a)$ if either $v \in V_1$, or $v \in V_2$ and there exists a P2's maximal sure winning strategy $\pi_2$ such that $\pi_2(v) = a$. Otherwise, $\widehat{\Delta}(v, a)$ is undefined.

\begin{theorem}
	\label{thm:attitude.agnostic}
	When P2 is not attitudinally cooperative, a strategy profile $(\pi_1, \pi_2)$ is a Nash equilibrium in $H$ if and only if $\pi_1$ is a maximal sure winning strategy for P1 in $\widehat{H}$ and $\pi_2$ is any valid P2 strategy in $\widehat H$.
\end{theorem}
%
%

The proof is similar to that of \refThm{thm:opposite.nash-equilibrium} and \refThm{thm:attitude.cooperative} and thus omitted.

\textbf{Both players need cooperation.} 
In this case, we characterize the Nash equilibrium in terms of the Pareto equilibrium. 
Intuitively, a Pareto equilibrium is a set of strategies where no player can achieve a better outcome by unilaterally changing their strategy without worsening the outcome for others \cite{wang1993existence}.

\begin{definition}
	\label{def:pareto}
	A strategy profile $(\pi_1, \pi_2)$ is a Pareto equilibrium if and only if the following conditions hold.
	\begin{enumerate}[i)]
		\item There does not exist a strategy $\pi_1' \in \Pi_1$ such that, for $j = 1, 2$,  $\rank_j(\pi_1', \pi_2) \leq \rank_j(\pi_1, \pi_2)$.
		
		\item There does not exist a strategy $\pi_2' \in \Pi_1$ such that, for $j = 1, 2$,  $\rank_j(\pi_1, \pi_2') \leq \rank_j(\pi_1, \pi_2)$.
	\end{enumerate}
	In both cases (i) and (ii), the inequality must hold strictly for at least one $j$.
\end{definition}

In the product game, the set of Pareto equilibria can be determined by computing the set of Pareto states in $V$, which is defined as the set,
\begin{align*}
	\mathsf{Pareto}(V) = \{v \in V \mid \nexists v' \in V: &\rank_1(v') < \rank_1(v) \text{ or } \\
			&\rank_2(v') < \rank_2(v)\}.
\end{align*}
It follows that a strategy profile $(\pi_1, \pi_2)$ is a Pareto equilibrium in $H$ if and only if the last state visited by the path $\Paths_H(v_0, \pi_1, \pi_2)$ is an element of the set $\mathsf{Pareto}(V)$.

%


%

\begin{lemma}
	\label{lma:partially.equal-ranks}
	For any two states $v, v' \in \mathsf{Pareto}(V)$, we have $\rank_1(v) = \rank_1(v')$ and $\rank_2(v) = \rank_2(v')$.
\end{lemma}
\begin{proof}
	By contradiction.
	Suppose there exist two states $v, v' \in \mathsf{Pareto}(V)$ such that $\rank_1(v) < \rank_1(v')$.
	Then, by definition of $\mathsf{Pareto}(V)$, $v'$ cannot be in $\mathsf{Pareto}(V)$.
\end{proof}

Intuitively, \refLma{lma:partially.equal-ranks} states that any two states in $\mathsf{Pareto}(V)$ have the same rank under $\calE_i$ for $i = 1, 2$.

When both players need cooperation, do they necessarily have an incentive to cooperate? 
For a player to cooperate, the resulting outcome from cooperation should be better than what the player can guarantee without cooperation. 
Formally, we say the players have an incentive to cooperate (in the instrumental sense) if the following inequality does not hold for either $i = 1$ or $i = 2$, $$\maxRank_i(\pi_i^{\max}) < \rank_i(\pi_1^P, \pi_2^P),$$ where $\pi_i^{\max}$ is a player-$i$'s maximal sure winning strategy and $(\pi_1^P, \pi_2^P)$ is a Pareto equilibrium.
If the inequality if true for either $i = 1$ or $i = 2$, then that player is guaranteed a better ranked outcome than it can achieve by cooperating. 
Hence, the player has no incentive to follow a Pareto equilibrium.

\begin{theorem} \label{thm:pareto}
	When players have an incentive to cooperate, a strategy $(\pi_1, \pi_2)$ is a Nash equilibrium if and only if it is a Pareto equilibrium. 
\end{theorem}
\begin{proof}
	When players have an incentive to cooperate, the set of outcomes is limited to Pareto states, $\mathsf{Pareto}(V)$.
	By \refLma{lma:partially.equal-ranks}, all states in $\mathsf{Pareto}(V)$ have the same rank. 
	In other words, neither player benefits by deviating from their Pareto equilibrium. 
\end{proof}

When no player has an incentive to cooperate, the Nash equilibria can be determined using \refThm{thm:attitude.cooperative} and \refThm{thm:attitude.agnostic} depending on whether the player who does not have an incentive to cooperate is attitudinally cooperative or not. 
This approach also applies when the $\mathsf{Pareto}(V)$ set is empty, which means that mutually beneficial cooperation is not possible in the current game.

\section{Experiment}
\label{sec:experiment}
We illustrate the application of the theoretical results using a drone delivery scenario depicted in \refFig{fig:gw-sim}.
The environment is a $5 \times 5$ gridworld featuring two drones, A and B, which must transport three packages from locations $p_1$, $p_2$, and $p_3$ to their respective destinations $d_1$, $d_2$, and $d_3$ while maximally satisfying their preferences over delivery schedules.
The drones can navigate in four compass directions: \texttt{N, E, S, W}. 
When a drone is in a cell containing package, \ie, a cell labeled $p_1, p_2, p_3$ in \refFig{fig:gw-sim}, they may collect the package using $\mathtt{pick}$ action.
When both drones occupy the $9$-neighboring cells, they can exchange package-$i$ if they have it using the $\mathtt{give}_i$ action where $i = 1, 2, 3$.
When both drones occupy the same cell, they can damage the other drone using $\mathtt{attack}$ action. 
The gridworld contains walls and obstacles that are bouncy, meaning if a drone's action leads it outside the grid's boundaries or into an obstacle, it returns to the cell from which it initiated the action.
The drones operate in turns, with each move action costing one unit of time while the actions $\mathtt{pick}, \mathtt{give}_i, \mathtt{attack}$ are instantaneous. 
Both drones must complete their tasks within a given time limit $T_{\max{}}$.

The key design question we ask is: Given that drone A is located at $(0, 0)$, where to place drone B for it to maximally satisfy its preference assuming both drones follow a Nash equilibrium? 

\subsection{Aligned Preferences} 
Consider the case when player preferences are aligned.
In practice, this may correspond to the situation when A and B are drones controlled by the same company .
Here, we expect the placement of drone B to maximally satisfy the preferences of both drones.

We consider the environment shown in \refFig{fig:gw-sim-a} with $T_{\max} = 10$. 
We express the preferences of both drones over four outcomes: $\varphi_i = \Eventually d_i$ for $i = 1, 2, 3$ and $\varphi_4 = \Eventually d_1 \land \Eventually (d_2 \lor d_3)$. 
The preferences are represented by the following \prefltlf~formula, which states that the drones prefer delivering the package at $p_1$ and at least one of those at $p_2$ and $p_3$ instead of delivering a only one package.
\begin{align*}
	\psi = 
	(\varphi_4 \strictpref \varphi_1) \prefAnd 
	(\varphi_4 \strictpref \varphi_2) \prefAnd 
	(\varphi_4 \strictpref \varphi_3) 
\end{align*}
It is assumed that delivering at least one package is strictly preferred to delivering none.

The preference automaton for $\psi_1$ is shown in \refFig{fig:aligned.pref-aut}. 
Specifically, sub-figure (a) shows the semi-automaton component of the preference automaton that tracks the progress made towards completion of various objectives in $\Phi$, while sub-figure (b) shows a preference graph that encodes the preorder $E$ on the states of the semi-automaton. 
The nodes of preference graph represent equivalence classes of semi-automaton states.  
An edge from one node to another in the preference graph represents that all semi-automaton states belonging to the partition defined by the latter node are strictly preferred to all states belonging to the partition defined by the former.
For example, consider a path in the game where B first delivers package 2 and then A delivers package 1.
When B delivers package 2, the semi-automaton transitions from its initial state $0$ to state $5$. 
Afterwards, when A delivers package 1, the semi-automaton transitions from state $5$ to state $6$. 
Similarly, consider a path where B first delivers package 3 and then A delivers package 2.
In this case, the semi-automaton transitions from state $0$ to $1$, and then from state $1$ to $2$.
To determine the preference between the two paths, we compare the nodes that define the partition containing the states $6$ and $2$ in the preference graph.
In this case, state $6$ corresponds to node $0$ and state $2$ corresponds to node $1$.
Since there exists a node $1$ to $0$, the first path is strictly preferred to the second.

\begin{figure}[tb]
	\centering
	\begin{multicols}{2}
		\centering
		\begin{subfigure}{\linewidth}
			\includegraphics[width=0.9\linewidth]{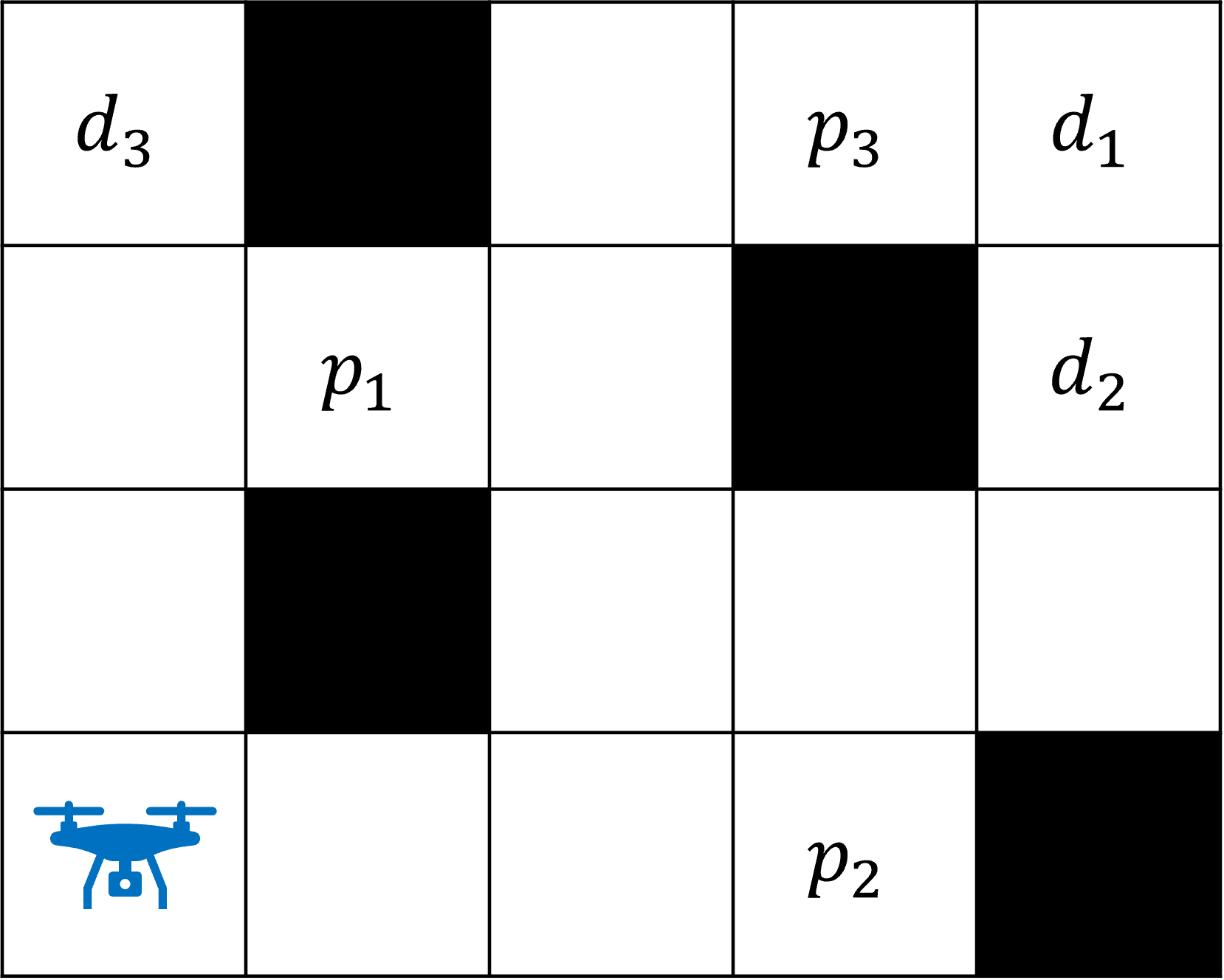}
			\caption{Scenario 1}
			\label{fig:gw-sim-a}
		\end{subfigure}
		\par
		\begin{subfigure}{\linewidth}
			\includegraphics[width=0.9\linewidth]{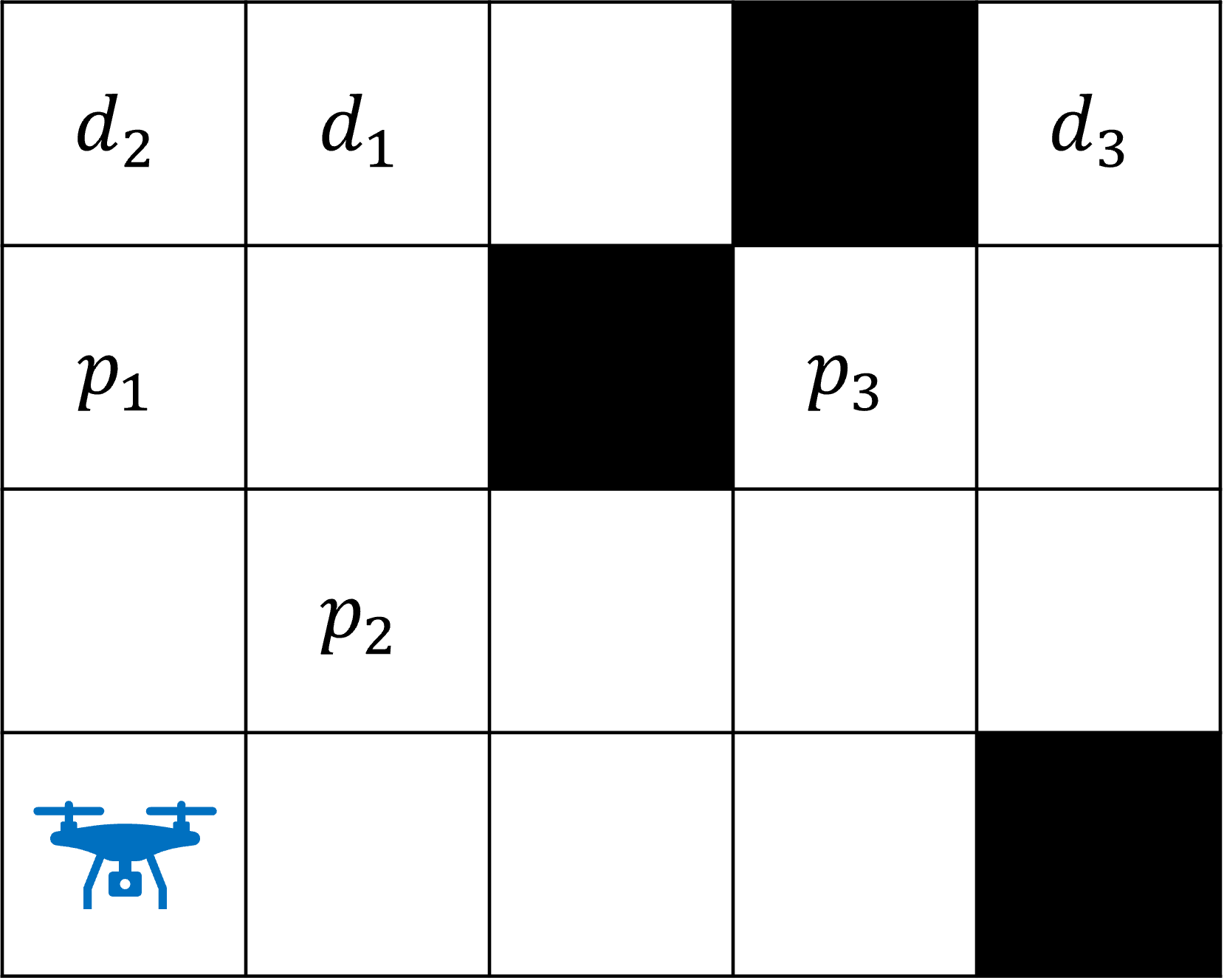}
			\caption{Scenario 2}
			\label{fig:gw-sim-b}
		\end{subfigure}
	\end{multicols}
	\caption{Two drone delivery environments.}
	\label{fig:gw-sim}
\end{figure}

\begin{figure}[tb]
	\centering
	\includegraphics[scale=0.3]{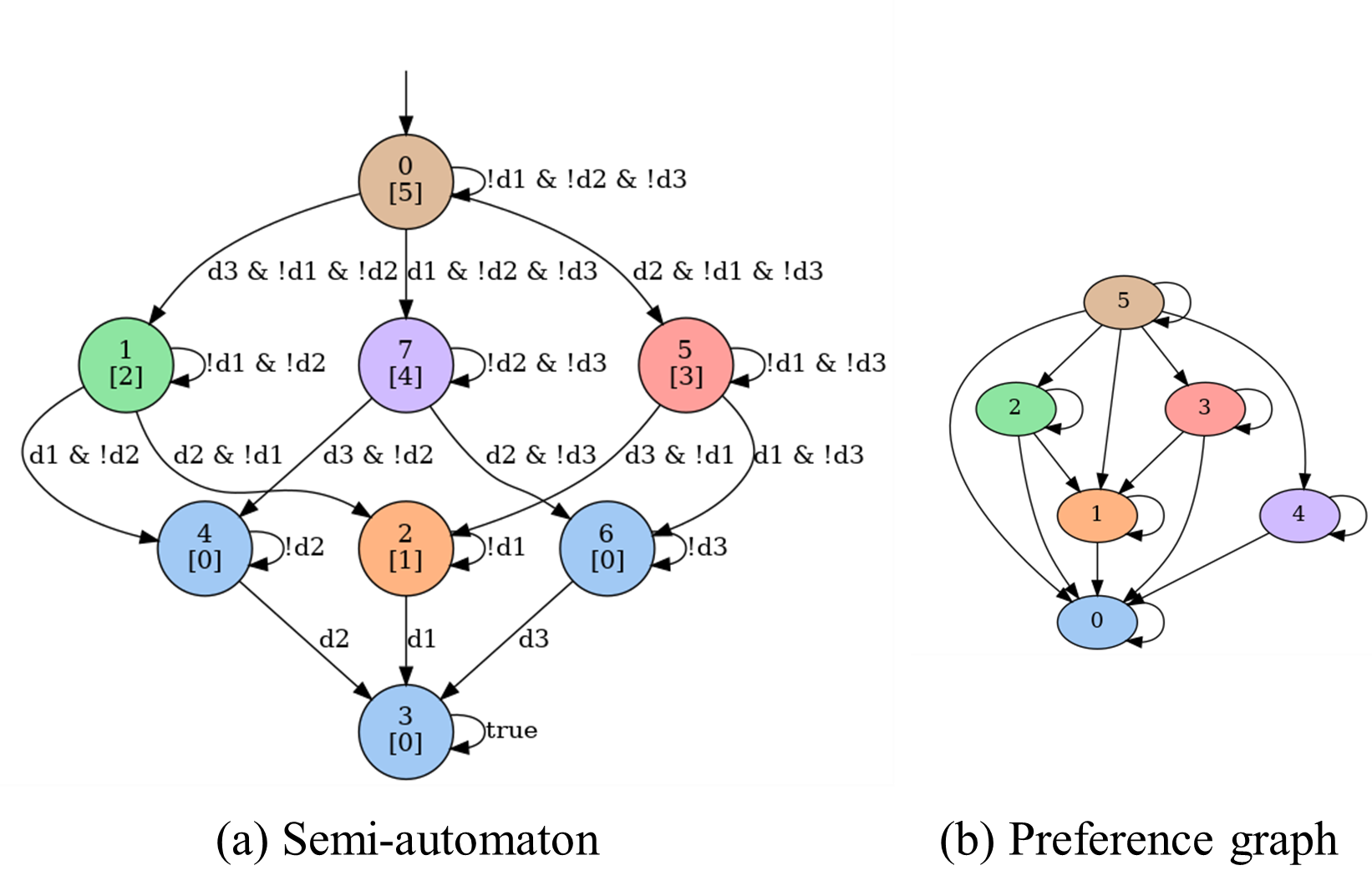}
	\caption{Preference automaton for \prefltlf~formula $\psi = (\varphi_4 \strictpref \varphi_1) \prefAnd (\varphi_4 \strictpref \varphi_2) \prefAnd (\varphi_4 \strictpref \varphi_3)$.}
	\label{fig:aligned.pref-aut}
\end{figure}


To determine the placement of drone B, we compute the rank of a maximal reachable state that can be visited when B starts from each eligible cell.
These ranks are shown in \refFig{fig:aligned.ranks}. 
The rank $-1$ (shown in black) depicts that B cannot start from that cell.

\begin{figure}[tb]
	\centering
	\includegraphics[scale=0.45]{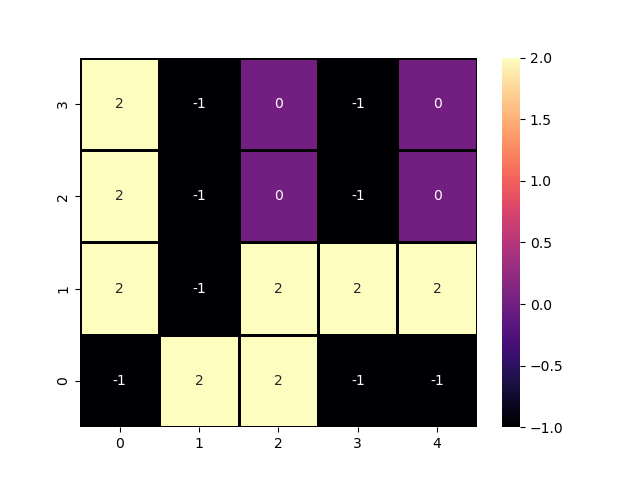}
	\caption{The rank of maximal reachable state in scenario 1.}
	\label{fig:aligned.ranks}
\end{figure}

We observe that if B is placed at a cell among $(2, 2), (2, 3), (4, 2, (4, 3)$, the drones can satisfy $\varphi_4$, which is their most preferred objective. 
To satisfy $\varphi_4$, they must coordinate their strategies.
For instance, when B starts at $(2, 3)$, the Nash equilibrium strategy of B requires it to pick package 3 and return to cell $(2, 3)$.  
Whereas, that of A requires it to pick package 1 by visiting cell $(1, 2)$. 
At this stage, they exchange package 1 and package 3, and then A delivers package 3 by visiting $(0, 3)$ and B delivers package 1 by visiting $(4, 3)$. 
Under this Nash equilibrium, $\varphi_4$ can be satisfied within $9$ time steps.

\begin{figure}[tb]
	\centering
	\includegraphics[scale=0.45]{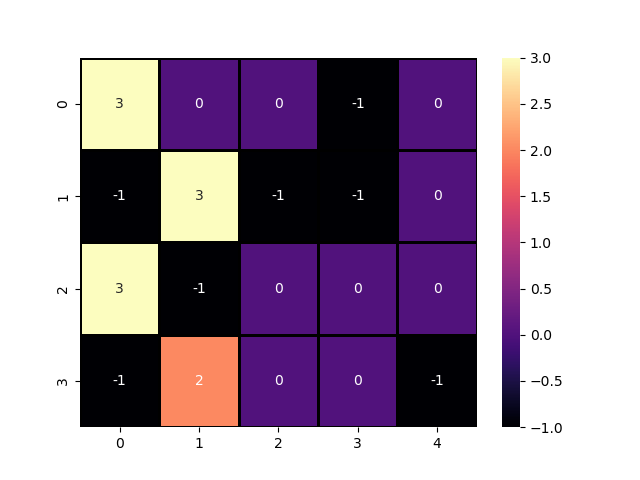}
	\caption{The smallest rank achievable by drone A by following a maximal sure winning strategy in scenario 2.}
	\label{fig:opposite.ranks}
\end{figure}

For case when B starts at $(2, 3)$ also illustrates that the Nash equilibrium is not unique. 
There exists another Nash strategy for B under which the exchange of package 1 and package 2 takes place when A is at $(1, 2)$ and B is at $(2, 2)$. 
In this case, $\varphi_4$ ise satisfied in $10$ time steps.

On the other hand, if B starts from any cell where the maximal reachable rank is equal to $2$, the drones can either package 2 or package 3. 
This is because it takes at least $11$ steps for the drones to pick and deliver any two packages from these initial states.

Therefore, we conclude that drone B must start at one of the cells among $(2, 2), (2, 3), (4, 2), (4, 3)$ to achieve the best possible outcome. 
Since the player preferences are aligned, the maximal reachable state for drone A also has the same rank as shown in \refFig{fig:aligned.ranks}.

\subsection{Completely Opposite Preferences} 

Consider the environment shown in \refFig{fig:gw-sim-b} with $T_{\max} = 10$. 
We express the preferences of both drones over four outcomes: $\varphi_i = \Eventually d_i$ for $i = 1, 2, 3$ and $\varphi_4 = \Eventually d_2 \land \Eventually d_3$. 
The preference of drone A is represented by the following \prefltlf~formula, which states that the A prefers delivering the package 1 to delivering only package 2 or only package 3, and delivering both packages 2 and 3 over delivering only package 2 or only package 3. 
\begin{align*}
	\psi_1 = 
	(\varphi_1 \strictpref \varphi_2) \prefAnd 
	(\varphi_1 \strictpref \varphi_3) \prefAnd 
	(\varphi_4 \strictpref \varphi_2) \prefAnd 
	(\varphi_4 \strictpref \varphi_3) 
\end{align*}
It is assumed that delivering at least one package is strictly preferred to delivering none.
The preference automaton for $\psi_1$ is shown in \refFig{fig:opposite.pref-aut}.
Since the preferences of drone B are completely opposite to that of drone A, we have $\psi_2 = (\varphi_2 \strictpref \varphi_1) \prefAnd (\varphi_3 \strictpref \varphi_1) \prefAnd (\varphi_2 \strictpref \varphi_4) \prefAnd (\varphi_3 \strictpref \varphi_4)$.  
Moreover, for drone B, delivering no package is strictly preferred to delivering at least one package.

\begin{figure}[tb]
	\centering
	\includegraphics[scale=0.26]{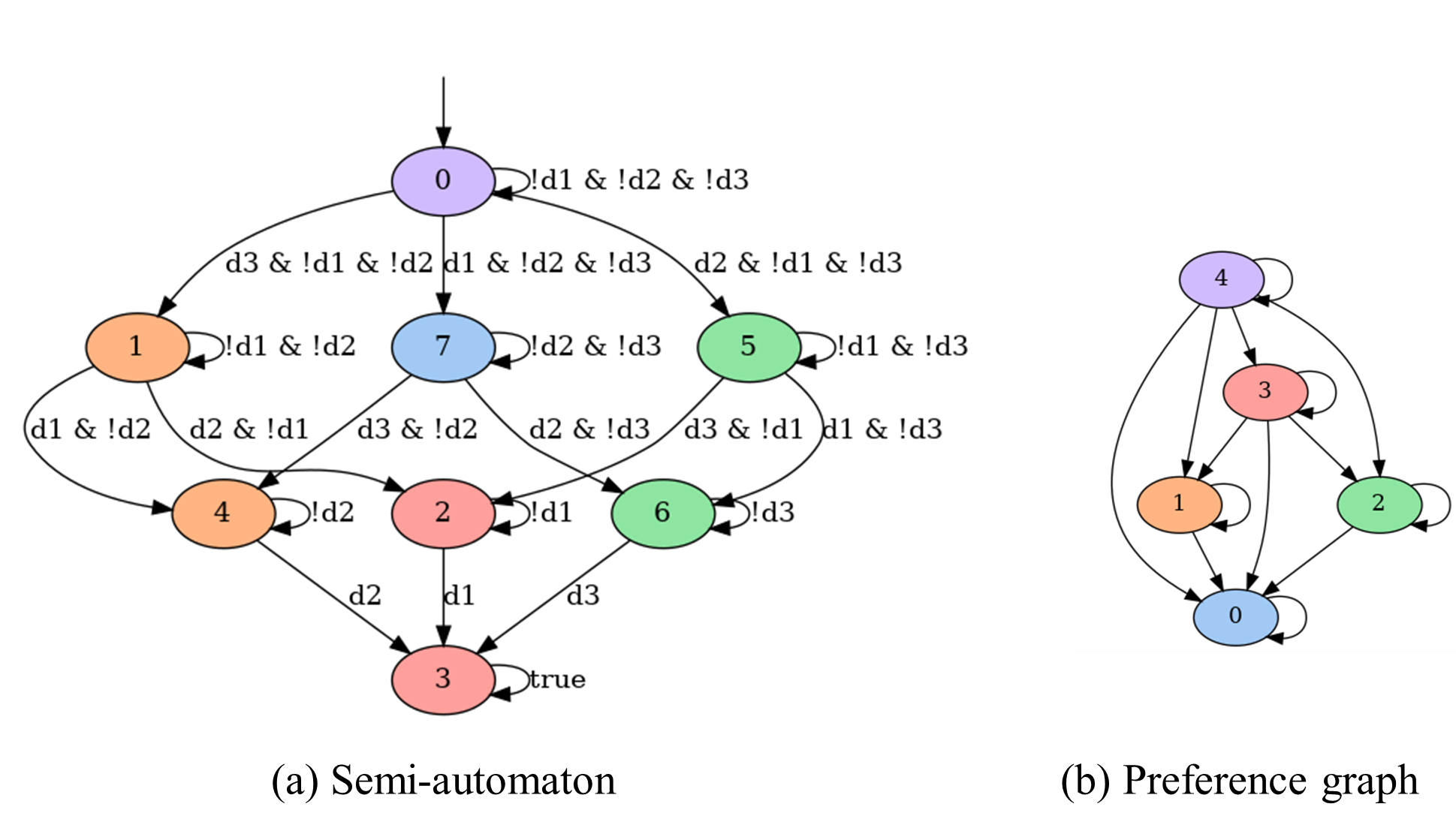}
	\caption{Preference automaton for \prefltlf~formula $\psi_1 = (\varphi_1 \strictpref \varphi_2) \prefAnd (\varphi_1 \strictpref \varphi_3) \prefAnd (\varphi_4 \strictpref \varphi_2) \prefAnd (\varphi_4 \strictpref \varphi_3)$.}
	\label{fig:opposite.pref-aut}
\end{figure}

\refFig{fig:opposite.ranks} shows in each cell the smallest rank that drone A can guarantee to achieve regardless of the strategy employed by drone B, when B starts from that cell.
For instance, the value $3$ at cells $(0, 1)$, $(0, 3)$, or $(1, 2)$ denotes that drone A cannot deliver any packages if B starts from any of these cells.
Specifically, in this case, drone B has a strategy to prevent A from either picking or dropping packages. 
For instance, if drone B starts at the cell $(0, 1)$, it has a strategy to prevent drone A from picking up any package.
This is because B can reach the cells labeled $p_1, p_2$ and $p_3$ before drone A can reach them, and A cannot enter the cell with B since B can use the $\mathtt{attack}$ action to disable it.
However, when B cannot prevent A from picking up any of the three packages, A can enforce a rank $0$ outcome against every possible strategy of B.

Therefore, we conclude that drone B must start at either $(0, 1)$, $(0, 3)$, or $(1, 2)$ to achieve the best possible outcome for itself.

\subsection{Partially Aligned Preferences}

We consider the environment shown in \refFig{fig:gw-sim-b} with $T_{\max} = 10$ and the same four outcomes from the previous subsection: $\varphi_i = \Eventually d_i$ for $i = 1, 2, 3$,  $\varphi_4 = ((\neg d_1 \land \neg d_3) \until d_2) \land \Eventually (d_2 \land \Eventually d_1)$, and $\varphi_5 = (\neg d_2 \land \neg d_3) \until d_1$. 
The preference of drone A and B are as follows.
\begin{align*}
	\psi_1 = 
	(\varphi_5 \strictpref \varphi_4) \prefAnd 
	(\varphi_4 \strictpref \varphi_1) \prefAnd 
	(\varphi_4 \strictpref \varphi_2) \prefAnd 
	(\varphi_4 \strictpref \varphi_3) \\ 
	%
	\psi_2 = 
	(\varphi_4 \strictpref \varphi_5) \prefAnd 
	(\varphi_5 \strictpref \varphi_1) \prefAnd 
	(\varphi_5 \strictpref \varphi_2) \prefAnd 
	(\varphi_5 \strictpref \varphi_3). 
\end{align*}
In words, drone A prefers delivering package 1 before any other package. 
If this is not possible, then it prefers to first deliver package 2 and then package 1.
If neither of above two specifications are possible, then drone A prefers delivering at least one package. 
On the other hand, drone B has preference opposite to A about the sequence of delivering packages 1 and 2. 
The preference automata for $\psi_1, \psi_2$ are not included due to space limitation. 

In \refFig{fig:partially_aligned.cooperation}, each cell contains a tuple that depicts whether drone A and B need cooperation when B starts from that cell. 
This tuple is determined using \refDef{def:need-of-cooperation}, by comparing the maximum rank outcome a drone can enforce with the minimum rank outcome the drone may achieve if the other drone was altruistic. 
For example, the value $(F, T)$ in cell $(2, 0)$ denotes that drone B needs cooperation while A does not.



Consider the case when drone B starts at $(2, 3)$. 
Here, B needs cooperation because it cannot pick and deliver any package by itself within $10$ steps.
However, A does not need cooperation because it can surely deliver package 1, thereby satisfying its most preferred outcome. 
In this scenario, if A was not attitudinally cooperative, the rank that B achieves is $2$. 
On the contrary, if A was attitudinally cooperative, then the Nash equilibrium requires A to pick package 1 and B to pick package 2. 
When the A is at $(1, 2)$ and B is at $(1, 1)$, B gives package 2 to A. 
Then, A delivers package 1 first, and then delivers package 2.
In this way, A achieves its most preferred outcome while also helping B to satisfy a rank $1$ outcome.

Now, consider the situation where drone B is at $(2, 1)$. 
Here, both drones need cooperation because, for A to achieve its best outcome, B must not deliver package 3 before A delivers package 1.
In fact, the rank of outcomes ensured by the maximal sure winning strategies of both drones is $2$. 
Therefore, we use \refThm{thm:pareto} to determine the Nash equilibrium. 
The set of Pareto states include states where drone A achieves rank $1$ and drone B achieves rank $0$.
For this, the drones must coordinate their strategies to deliver package 2 first and then package 1. 
Specifically, drone A must choose $E$ at its first step, and allow B to first pick package 2 and then package 1. 
The drone B then visits $(0, 3)$ to deliver package 2 and then gives package 1 to A, who is at $(1, 3)$ to deliver package 1.
The Pareto equilibrium does not allow A to move north at $(0, 0)$ since it has a sure winning strategy to deliver package 1 in this case.

Since drone B needs cooperation from all initial positions, it should start from the state $(2, 1)$, where the Pareto strategy ensures the most preferred outcome for B. 
If this were not the case, then B should consider A's attitude and determine what is the best rank it can achieve if it started from each eligible cell and select the smallest.

\begin{figure}[t]
	\centering
	\includegraphics[scale=0.35]{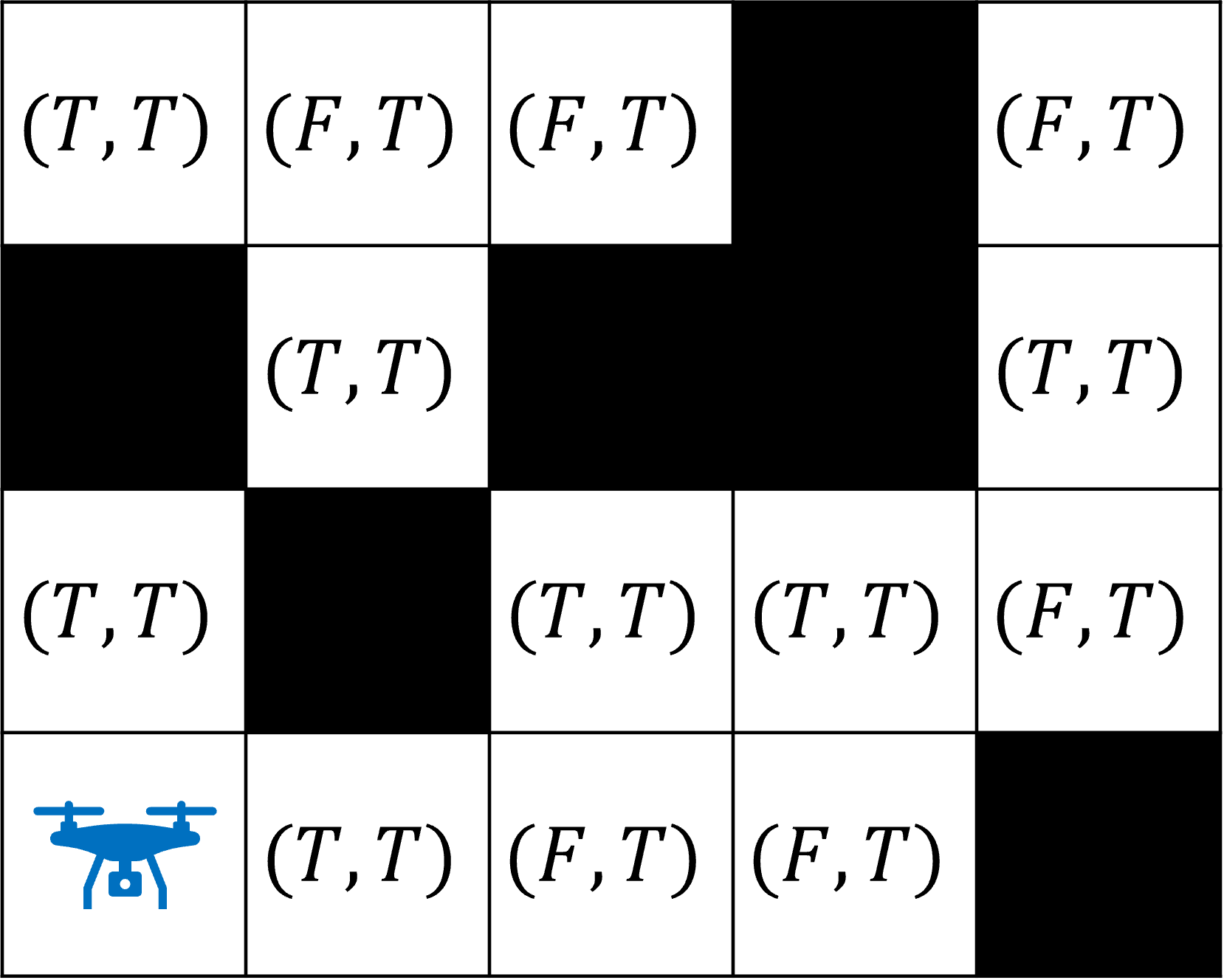}
	\caption{Need for cooperation in scenario 2. The tuple in each cell denotes whether drones A and B need cooperation.}
	\label{fig:partially_aligned.cooperation} 
\end{figure}

\section{Conclusion} 
We studied the problem of characterizing Nash equilibrium in deterministic two-player turn-based games on graphs where players aim to maximally satisfy their preference over \ac{ltlf} formulas.
We developed an automata-theoretic approach to computing the set of Nash equilibria under various scenarios of preference alignment: fully aligned, partially aligned, and completely opposite.
We demonstrated that player attitudes affects the Nash equilibria when player preferences are partially aligned, thereby gaining key insights into when a player needs cooperation, and when a player has instrumental or attitudinal incentive to cooperate. 
We also established the existence of Nash equilibria in all scenarios.

However, the current study only considered the deterministic turn-based games on graphs. 
Several applications in robotics and AI are probabilistic in nature and involve concurrent interactions between players.
In future, we will study the characterization of Nash equilibria for stochastic concurrent games on graphs. 


\bibliographystyle{plain}        
\bibliography{sample}           



\end{document}